\documentclass[11pt,a4paper,twoside]{article}
\usepackage{a4wide, amssymb, amsmath, amsthm, graphics, comment, xspace, enumerate}
\usepackage{graphicx,multirow}
\usepackage[a4paper,colorlinks=true,citecolor=blue,urlcolor=blue,linkcolor=blue
                   , bookmarksopen=true,unicode=true,pdffitwindow=true]{hyperref}
\usepackage[noend]{algpseudocode}
\usepackage{caption}
\usepackage{tabularx}
\usepackage{enumitem, multirow}
\usepackage[usenames,dvipsnames]{color}
\usepackage{alltt}
\usepackage{color}
\usepackage{listings}
\usepackage{cite}

\usepackage{float}
\usepackage{xspace}
\usepackage{algorithmicx}
\usepackage[section]{algorithm}
\usepackage[noend]{algpseudocode}

\definecolor{string}{rgb}{0.7,0.0,0.0}
\definecolor{comment}{rgb}{0.13,0.54,0.13}
\definecolor{keyword}{rgb}{0.0,0.0,1.0}
\usepackage{amsmath,amsfonts,amssymb,subfigure}
\usepackage{tikz}
\usetikzlibrary{arrows}
\usetikzlibrary{decorations}
\usetikzlibrary{patterns}

\subfiguretopcaptrue
\tikzstyle{vtx}=[circle, inner sep= 0pt, minimum size= 1.2mm, fill]

\hypersetup{pdftitle={Efficient computation of trees with minimal atom-bond connectivity index II}}

\newtheorem{te}{Theorem}[section]
\newtheorem{pro}[te]{Proposition}
\newtheorem{de}{Definition}[section]

\newtheorem{co}[te]{Corollary}
\newtheorem{lemma}[te]{Lemma}

\newtheorem{observaton}{Observation}[section]

\newcommand{\beq}{\begin{eqnarray}}
\newcommand{\eeq}{\end{eqnarray}}

\newcommand{\beqs}{\begin{eqnarray*}}
\newcommand{\eeqs}{\end{eqnarray*}}

\newcommand{\ABC}{{\rm ABC}}

\allowdisplaybreaks

\begin{document}
%
\title{Forbidden branches in trees \\ with minimal atom-bond connectivity index}
\maketitle
\begin{center}
{\large \bf Darko Dimitrov$^a$, Zhibin Du$^b$, Carlos M. da Fonseca$^{c,d}$}
\end{center}

\baselineskip=0.20in
\begin{center}
{\it $^a${Hochschule f\"ur Technik und Wirtschaft Berlin, Germany \& \\
 Faculty of Information Studies, Novo Mesto, Slovenia}}
\\E-mail: {\tt darko.dimitrov11@gmail.com}
\end{center}
\baselineskip=0.20in
\begin{center}
{\it $^b$School of Mathematics and Statistics, Zhaoqing University
\\ Zhaoqing 526061, China}
\\E-mail: {\tt zhibindu@126.com}
\end{center}
\baselineskip=0.20in
\begin{center}
{\it $^c$Department of Mathematics, Kuwait University
\\ Safat 13060, Kuwait}
\\E-mail: {\tt carlos@sci.kuniv.edu.kw}
\end{center}
\baselineskip=0.20in
\begin{center}
{\it $^d$Department of Mathematics, University of Primorska \\
Glagoljsa\v ska 8, 6000 Koper, Slovenia}
\\E-mail: {\tt carlos.dafonseca@famnit.upr.si}
\end{center}

\baselineskip=0.20in
\vspace{6mm}
\begin{abstract}
The atom-bond connectivity (ABC) index has been, in recent years, one of the most
actively studied vertex-degree-based graph invariants in chemical
graph theory. For a given graph $G$, the ABC index is defined as
$\sum_{uv\in E}\sqrt{\frac{d(u) +d(v)-2}{d(u)d(v)}}$, where
$d(u)$ is the degree of vertex $u$ in $G$ and $E(G)$ denotes the set
of edges of $G$. In this paper we present some new structural properties
of trees with a minimal ABC index (also refer to as a minimal-ABC tree), which is a step further towards understanding
their complete characterization. We show that a minimal-ABC tree
cannot simultaneously contain  a $B_4$-branch and $B_1$ or
$B_2$-branches.
\end{abstract}
{\small \hspace{0.25cm} \textbf{Keywords:} Atom-bond connectivity
index; tree; extremal graphs}
%
%
\medskip
%
%
%
%
\section[Introduction and preliminaries]{Introduction and preliminaries}

Let $G=(V, E)$ be a simple undirected graph with $n$ vertices.
For $v \in V$, the degree of $v$, denoted by $d(v)$, is the
number of edges incident to $v$. In 1998, Estrada, Torres, Rodr{\' i}guez and Gutman
\cite{etrg-abc-98} proposed a vertex-degree-based graph
topological index - {\em the atom-bond connectivity (ABC) index} -
defined as
\beq \ABC(G)=\sum_{uv\in E} f(d(u), d(v)), \nonumber
\eeq where
\beq \label{eqn:001}
f(d(u), d(v))=\sqrt{\frac{d(u)
+d(v)-2}{d(u)d(v)}}. \noindent
\eeq
It was shown that the ABC index
can be a valuable predictive tool in the study of heat
formation in alkanes. Ten years later, Estrada elaborated in
\cite{e-abceba-08} an innovative quantum-theory-like explanation of
this topological index.
Incontestably, this topic has triggered
tremendous interest in both mathematical
and chemical research communities, leading to a  number of
results that incorporate the structural properties and the computational aspects of the graphs with
extremal properties \cite{adgh-dctmabci-14, ahs-tmabci-13, ahz-ltmabci-13, cg-eabcig-11,
clg-subabcig-12, dmga-cbabcig-16, d-abcig-10, d-sptmabci-17, d-etrabci-17, ddf-sptmabci-3-2015, dis-rmabciggp-17, dgf-abci-11, dgf-abci-12,
d-abcircg-2015, fgv-abcit-09, gf-tsabci-12, ghl-srabcig-11, gs-sabcitnpv-16, ftvzag-siabcigo-2011, gfi-ntmabci-12, gfahsz-abcic-2013,
lcmzczj-twmabciatgnl-16, llgw-pcgctmabci-13, vh-mabcict-2012,
xz-etfdsabci-2012, xzd-abcicg-2011, xzd-frabcit-2010,
zc-rbabcfgai-2015, fgiv-cstmabci-12,
d-ectmabci-2013, lccglc-fcstmaibtds-14, lmccgc-pstmabci-15}. On the other hand, the physico-chemical
applicability of the ABC index has also been confirmed and extended in
several other studies \cite{cll-abcbsp-13, gg-nwabci-10, gtrm-abcica-12, k-abcibsfc-12,
yxc-abcbsp-11}.

It has been proven that deleting/adding an edge in a graph strictly
decreases/increases its ABC index \cite{cg-eabcig-11,dgf-abci-11}.
Consequently, among all connected graphs, a tree/the complete graph
has minimal/maximal ABC index.

It has been shown that among the trees of a given order, the star is
the one with maximal ABC index \cite{fgv-abcit-09}. Notwithstanding,
a thoroughgoing characterization of trees with minimal ABC index,
also referred to as minimal-ABC trees, still remains an open problem. This
paper represents a step further towards the comprehensive classification of
such trees.

In the sequel, we present some additional results and main notations
that will be used throughout the paper. A vertex of degree one is a
{\it pendant vertex}.
As in
\cite{gfi-ntmabci-12}, a sequence of vertices of a graph $G$,
$S_k=v_0 \, v_1 \cdots v_k$,  will be  called a {\it pendant path} if
each two consecutive vertices in $S_k$ are adjacent in $G$,
$d(v_0)>2$, $d(v_i)=2$, for $i=1,\dots, k-1$, and $d(v_k)=1$. The
length of the pendant path $S_k$ is $k$. If $d(v_k) > 2$, then $S_k$
is called an {\em internal path} of length $k$.

A $B_1$-branch is a path of length $2$ attached to a vertex that has at least one child of degree at least 3. We call the vertex of degree $2$ in a $B_1$-branch as the center of such $B_1$-branch. A $B_k$-branch, for $k \geq
2$, is a (sub)graph comprised of vertex $v$ of degree $k+1$ and
$k$ pendant paths of length $2$ that all have $v$ as a common
vertex. We call the vertex $v$ also the {\em center} of the $B_k$-branch.
Moreover, a $B_k^*$-branch, for $k \geq
1$, is a (sub)graph obtainable from $B_k$ by attaching an additional vertex to a pendant vertex of $B_k$-branch.



A $k$-{\em terminal vertex} of a rooted tree is a vertex of degree $k+1 \ge 3$, which is a parent of only $B_{\geq 1}$-branches,
such that  at least one branch among them is a $B_1$-branch (or $B^*_1$-branch).
The (sub)tree, induced by a  $k$-{\em terminal vertex} and all its (direct and indirect) children vertices, is called a $k$-{\em terminal branch} or
$T_k$-{\em branch}.

A sequence $D=(d_1, d_2, \dots, d_n)$ is {\it graphical} if there is
a graph whose vertex degrees are $d_i$, $i=1,\dots,n$. Additionally,
if $d_1 \geq d_2\geq \dots \geq d_n$, then  $D$ is called a {\it degree sequence}.

To determine the minimal-ABC trees with an order less than $10$ is a
simple task, therefore to simplify the exposition in the rest of the paper,
we assume that all the trees are of an order at least $10$.

In $2008$, Wang \cite{w-etwgdsri-2008} defined a {\em greedy tree}
as follows.

\begin{de}\label{def-GT}
Suppose the degrees of the non-leaf vertices are given, the greedy
tree is achieved by the following `greedy algorithm':
\begin{enumerate}
\item Label the vertex with the largest degree as $v$ (the root).
\item Label the neighbors of $v$ as $v_1, v_2,\dots,$ assign the largest degrees available to them such that $d(v_1) \geq d(v_2) \geq
\cdots$.
\item Label the neighbors of $v_1$ (except $v$) as $v_{11}, v_{12}, \dots$ such that they take all the largest
degrees available and that $d(v_{11}) \geq d(v_{12}) \geq\cdots$
then do the same for $v_2, v_3,\dots$.
\item Continue the labeling of the unlabeled vertices in the same manner as in 3.,
always starting with the neighbors of the labeled vertex with the largest degree.
\end{enumerate}
\end{de}


Next we present  the so-called {\em switching transformation}
explicitly stated by Lin, Gao, Chen, and Lin
\cite{lgcl-mabcicggds-13}. The switching transformation was used in
the proofs of some characterizations of the minimal-ABC trees, as it
was the case with Lemma \ref{lemma-DS} by  Gan, Liu, and You in
\cite{gly-abctgds-12}.

\begin{lemma}[Switching transformation]\label{lemma-switching}
Let $G=(V,E)$ be a connected graph
with $uv, xy \in E(G)$ and $uy, xv \notin E(G)$.
Let $G_1= G - uv -xy + uy + xv$. If $d(u) \geq d(x)$ and $d(v) \leq d(y)$, then
$\ABC(G_1) \leq \ABC(G)$, with the equality if and only if $d(u) = d(x)$ or $d(v) =d(y)$.
\end{lemma}

\begin{lemma}\label{lemma-DS}
In a minimal-ABC tree, every path $v_0 v_1 \cdots v_t v_{t+1}$, where $v_0$ and $v_{t+1}$ are leaves, has the properties:
\begin{enumerate}
\item if  $t$ is odd, then
$$
d(v_1) \leq d(v_t) \leq d(v_2) \leq  d(v_{t-1}) \leq \cdots \leq d(v_{\frac{t-1}{2}}) \leq d(v_{\frac{t+3}{2}}) \leq d(v_{\frac{t+1}{2}});
$$
\item if  $t$ is even, then
$$
d(v_1) \leq d(v_t) \leq d(v_2) \leq  d(v_{t-1}) \leq \cdots \leq d(v_{\frac{t+4}{2}}) \leq d(v_{\frac{t}{2}}) \leq d(v_{\frac{t+2}{2}}).
$$
\end{enumerate}
\end{lemma}

Almost simultaneously, Xing and Zhou \cite{xz-etfdsabci-2012},
Gan, Liu, and You \cite{gly-abctgds-12} and Lin, Gao, Chen, and Lin
\cite{lgcl-mabcicggds-13} independently gave  the following
characterization.

\begin{te}\label{thm-DS}
Given the degree sequence, the greedy tree minimizes the ABC index.
\end{te}



In \cite{gfi-ntmabci-12}, Gutman, Furtula, and Ivanovi{\' c}
obtained the following result.

\begin{te}\label{thm-GFI-15}
A minimal-ABC tree with $n \geq 10$ vertices contains neither
internal paths of any length $k \geq 2$ nor pendant paths of length
$k \geq 4$.
\end{te}

An immediate, but important, consequence of Theorem \ref{thm-GFI-15}
is the next corollary \cite{gfi-ntmabci-12, xz-etfdsabci-2012}.

\begin{co}\label{co-GFI-10}
Let $T$ be a minimal-ABC tree. Then the subgraph induced
by the vertices of $T$ whose degrees are greater than two is also a
tree.
\end{co}

An improvement of Theorem \ref{thm-GFI-15} is the following result
by Lin, Lin, Gao, and Wu \cite{llgw-pcgctmabci-13}.

\begin{te}\label{thm-LG-10}
Each pendant vertex of a minimal-ABC tree belongs to a pendant path of length $k$, where $k = 2$ or $3$.
\end{te}

Further improvement of Theorem \ref{thm-GFI-15} was given by Gutman,
Furtula, and Ivanovi{\' c}.

\begin{te}[\cite{gfi-ntmabci-12}] \label{thm-GFI-30}
The $n$-vertex tree, $n \geq 10$, with minimal ABC-index contains at
most one pendant path of length $3$.
\end{te}




By Theorem \ref{thm-LG-10} and Corollary \ref{co-GFI-10}, it
follows that the minimal-ABC trees can be obtained by unifying each
pendant vertex of a tree $T$ with the center vertex of a $B_k$-branch.
In particular, if $T$ is just a star, then the minimal-ABC trees are the
same trees that are minimal with respect to Kragujevac trees \cite{hag-ktmabci-14}.

In \cite{d-sptmabci-2-2015}  it was shown that a minimal-ABC tree contains at
most one
$T_k$-branch, $k \geq 2$. The next three results, which will be used
in the proofs in the next two sections, considering the bounds on the
number of $B_1$-branches, a (non)coexistence of some types of
$B_k$-branches that have a common parent vertex as well as some
conditions on the existence of pendant path of length $3$.






\begin{te}[\cite{d-sptmabci-2014}] \label{te-no5branches-10}
A minimal-ABC tree does not contain a $B_k$-branch, $k \geq 5$.
\end{te}

\begin{te}[\cite{d-sptmabci-2014}]  \label{thm-20}
A  minimal-ABC tree does not contain more than four $B_4$-branches.
\end{te}

\begin{te}[\cite{d-sptmabci-2-2015}]\label{thm-30}
A minimal-ABC tree $G$ contains at most four $B_1$-branches.
Moreover, if $G$ is a $T_k$-branch itself, then it contains at most
three $B_1$-branches.
\end{te}

\begin{lemma}[\cite{d-sptmabci-2014}] \label{lemma-10}
A minimal-ABC tree does not contain
\begin{enumerate}
\item[(a)] a $B_1$-branch and a $B_4$-branch,
\item[(b)] a $B_2$-branch and a $B_4$-branch,
\end{enumerate}
that have a common parent vertex.
\end{lemma}

\begin{lemma}[\cite{df-ftmabc-2015}] \label{lemma-10-1}
A minimal-ABC tree does not contain
a $B_3$-branch and a $B_1^*$-branch
that have a common parent vertex.
\end{lemma}

\begin{te}[\cite{ddf-sptmabci-3-2015}] \label{te-B2-10}
A minimal-ABC tree  of order $n > 18$ with a pendant path of length
$3$ may contain $B_2$-branches if and only if it is of order $161$
or $168$. Moreover, in this case the minimal-ABC tree is comprised of a
single central vertex, $B_3$-branches and one $B_2$-branch, including a
pendant path of length $3$ that may belong to a $B_3^*$- or
$B_2^*$-branch.
\end{te}





For further properties of the minimal-ABC trees, the reader is
referred to \cite{d-ectmabci-2013, d-sptmabci-2014,
d-sptmabci-2-2015, ddf-sptmabci-3-2015, d-abcircg-2015,
df-ftmabc-2015}.  Before we proceed with the main results of this
paper, we present the following observation that will be applied in
the further analysis.


\begin{observaton}  \label{obs-DS-branches1}
Let $G$ be a minimal-ABC tree with the root vertex $v_0$ and let $v_0, v_1, \dots, v_n$ be the sequence of vertices
obtain by the breadth-first search of $G$.
If $d(v_i), d(v_j) \geq 3$ and $i < j$, then by Lemma~\ref{lemma-switching}, we may assume that $d(v_i) \geq d(v_j)$.
\end{observaton}

In what follows, we will present some new structural properties of minimal-ABC trees
by strengthening Lemma~\ref{lemma-10}.
In Section~\ref{sec-B4B2} we will show that a minimal-ABC tree cannot contain a $B_4$-branch and a $B_2$-branch simultaneously,
while in Section~\ref{sec-B4B1} we will show that a minimal-ABC tree cannot contain a $B_4$-branch and a $B_1$-branch simultaneously.

\section{Trees containing simultaneously $B_4$- and  $B_2$-branches}\label{sec-B4B2}

This section is devoted to proving that a minimal-ABC tree cannot
contain a $B_4$-branch and a $B_2$-branch simultaneously. First we
state two technical lemmas which will be particularly useful
throughout the paper. For the proof the reader can be referred to
\cite{d-sptmabci-2014}. The function $f(x,y)$ is defined as in
\eqref{eqn:001}.

\begin{lemma} \label{appendix-pro-010}
Let $g(x,y)=-f(x,y)+f(x+\Delta x,y-\Delta y)$, with  real numbers
$x, y \geq 2$,  $\Delta x \geq 0$,  $0 \leq \Delta y < y$. Then
$g(x,y)$ increases in $x$ and decreases in $y$.
\end{lemma}

Due to the symmetry of $f(x,y)$, Lemma \ref{appendix-pro-010} can be
rewritten as follows.

\begin{lemma}  \label{appendix-pro-010-2}
Let $g(x,y)=-f(x,y)+f(x -\Delta x,y+\Delta y)$, with  real numbers
$x, y \geq 2$,  $\Delta y \geq 0$,  $0 \leq \Delta x < x$. Then
$g(x,y)$ decreases in $x$ and increases in $y$.
\end{lemma}

First, we present two crucial properties for the minimal-ABC trees, which will be applied in the proof of the main result
in this section - Theorem~\ref{thm-noB2B4}.

\begin{pro} \label{pro-B4-B2}
A minimal-ABC tree cannot contain a $B_4$-branch and a $B_2$-branch
simultaneously, where $u$ is the parent vertex of the center of a $B_4$-branch, and $v$ is the parent vertex of the center of a $B_2$-branch, when
\begin{itemize}
\item $d(v)=5$ and $d(v)  \leq d(u) \leq 12$;
\item $d(v)=6$ and $d(v)  \leq d(u) \leq 24$;
\item $d(v)=7$ and $d(v)  \leq d(u)$.
\end{itemize}
\end{pro}

\begin{proof}
Suppose that $G$ is a minimal-ABC tree that contains a $B_4$-branch and a $B_2$-branch
simultaneously.
Let $u_1$ be the center vertex of a $B_4$-branch, with the parent vertex $u$, and $v_1$ the center vertex of a $B_2$-branch, with the parent vertex $v$.
From Observation \ref{obs-DS-branches1}, we know that $u_1$ occurs before $v_1$ in the breadth-first search of $G$.
Furthermore, by Lemma~\ref{lemma-switching},
it follows that $d(u)\geq d(v) \geq d(u_1)
= 5$.

Here we apply the transformation $\mathcal{T}$ depicted in
Figure~\ref{noB4B2-10}.
\begin{figure}[h]
\begin{center}
\includegraphics[scale=0.90]{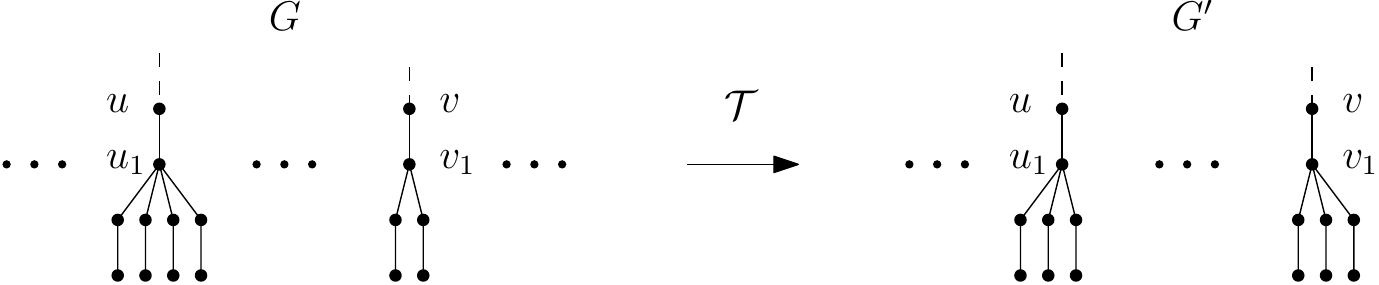}
\caption{ An illustration of the transformation $\mathcal{T}$ from
the proof of Theorem~\ref{thm-noB2B4}.} \label{noB4B2-10}
\end{center}
\end{figure}
After applying $\mathcal{T}$, the change of the ABC index of $G$ is
\beq \label{eq-lemma-B4-70}
ABC(G') - ABC(G) &=& g_1(d(u), d(v)) \nonumber\\
&=& -f(d(u), 5) + f(d(u), 4) -f(d(v), 3) + f(d(v), 4)\, . \nonumber
\eeq

By Lemma~\ref{appendix-pro-010}, $g_1(d(u), d(v))$ increases
in $d(u)$.
For $d(v)=5,6,7$, it can be verified  that $13,25, 67$, respectively, are the smallest values
of $d(u)$ for which $g_1(d(u), d(v))$ is non-negative.

For $d(v)=8$,
we obtain that
\beq \label{eq-lemma-B4-70-8}
g_1(d(u), 8)
&=& -f(d(u), 5) + f(d(u), 4) -f(8, 3) + f(8, 4) \nonumber\\
&\le& \lim_{d(u) \to \infty}  (-f(d(u), 5) + f(d(u), 4)) -f(8, 3) + f(8, 4) \nonumber\\
&=& -\sqrt{\frac{1}{5}}+\sqrt{\frac{1}{4}} -f(8, 3) + f(8, 4) \nonumber\\
&<& 0,
\eeq
i.e., the function $g_1(d(u), 8)$
is negative.

By Lemma~\ref{appendix-pro-010-2}, we have that
$g_1(d(u), d(v))$ decreases in $d(v)$, and thus also
$$
g_1(d(u), d(v)) \le g_1(d(u), 8) < 0
$$
for $d(v)>8$. 

Now the result follows.
\end{proof}

Applying the inverse transformation of the transformation $\mathcal{T}$ depicted in
Figure~\ref{noB4B2-10}, together with the proof of Proposition \ref{pro-B4-B2}, the following proposition is immediate.

\begin{pro} \label{pro-B3-B3}
A minimal-ABC tree cannot contain two $B_3$-branches, attached to two different vertices, simultaneously, where $u$ and $v$ are the two (different) parent vertices of the centers of two $B_3$-branches, when
\begin{itemize}
\item $d(v)=5$ and $d(u) \geq 13$;
\item $d(v)=6$ and $d(u) \geq 25$;
\item $d(v)=7$ and $d(u) \geq 67$.
\end{itemize}
\end{pro}

We now state and prove the main result of this section.

\begin{te} \label{thm-noB2B4}
A minimal-ABC tree cannot contain a $B_4$-branch and a $B_2$-branch
simultaneously.
\end{te}

\begin{proof}
Suppose that $G$ is a minimal-ABC tree that contains a $B_4$-branch and a $B_2$-branch
simultaneously. Let $u_1$ be the center vertex of a $B_4$-branch, with the parent vertex $u$ and the grandparent vertex
$u_p$ (if existed).
Let $v$ be the last vertex in the breadth-first search of $G$, which is a parent of a $B_2$-branch.
Denote by $v_1$ the center vertex of that $B_2$-branch and $v_p$ the parent of $v$.
Here  also, by Observation \ref{obs-DS-branches1}, we have $u_1$ occurs before $v_1$ in the breadth-first search of $G$.
Furthermore, by Lemma~\ref{lemma-switching},
it follows that $d(u)\geq d(v) \geq d(u_1) \geq d(v_1)$.
Notice that the existence of $B_4$- and $B_2$-branches in $G$ implies that every pendant path in $G$ is of length $2$, from Theorem \ref{te-B2-10}.

By Lemma~\ref{lemma-10}(b), $u_1$ and $v_1$ cannot have a common parent
vertex, i.e., $u \ne v$. And by Proposition \ref{pro-B4-B2}, the following cases remain to be considered:
\begin{itemize}
\item $d(v)=5$ and $d(u) \geq 13$;
\item $d(v)=6$ and $d(u) \geq 25$;
\item $d(v)=7$ and $d(u) \geq 67$.
\end{itemize}



By Lemma~\ref{lemma-10}(b), $v$ cannot be a parent vertex of a
$B_4$-branch. And we will show that $v$ cannot be a parent vertex of a
$B_3$-branch, either.

Suppose to the contrary that there is a $B_3$-branch attached to $v$.
From Proposition \ref{pro-B3-B3}, no $B_3$-branch is attached to $u$.

Suppose that every branch attached to $u$ is a $B_k$-branch, except the one containing $v$ (not necessarily exist such branch). By Lemma \ref{lemma-switching} and Lemma~\ref{lemma-10}(b), we know that every branch attached to $u$ is actually a $B_4$-branch. However, Theorem \ref{thm-20} claims that there are at most four $B_4$-branches in $G$, which is a contradiction to $d(u) \ge 13$.

Suppose that there is a branch attached to $u$ which is not a $B_k$-branch, denote by $w$ the child of $u$ in such branch.
If $v$ occurs before $w$ in the breadth-first search of $G$, then by Observation \ref{obs-DS-branches1}, there must have $B_2$-branch attached to $w$, which is a contradiction to the choice of $v$.
If $w$ occurs before $v$ in the breadth-first search of $G$ ($v$ is also a child of $u$), then by Observation \ref{obs-DS-branches1}, either there exists $B_3$-branch attached to $w$, which is a contradiction to Proposition \ref{pro-B3-B3}, or there are at least five $B_4$-branch in $G$, which is a contradiction to Theorem \ref{thm-20}.

So no $B_3$-branch is attached to $v$, i.e., every branch attached to $v$ is $B_1$- or $B_2$-branch.



Let $n_1$ and $n_2$ be the number of $B_1$- and
$B_2$-branches, correspondingly,
that have $v$ as the parent vertex. It holds that
$d(v)=n_1+n_2+1$, with $n_1 \geq 0$ and $n_2 \geq 1$.
Moreover, by Theorem~\ref{thm-30}, $n_1$ can be at most $4$.

In the rest of the proof, we will consider the remaining cases when
$d(v)=5,6,7$. Further, we distinguish six cases with respect to the
value of $n_2$: $n_2=1,2,3,4,5,6$.

In addition, notice that $v$ cannot occur before
$u$ in the breadth-first search of $G$, and recalling that $u \ne v$
from Lemma~\ref{lemma-10}(b), thus there are four possibilities with
respect to the relationship between vertices $u,v, u_p$ and $v_p$,
that we are interested in:
\begin{itemize}
\item[(a)] $u,v, u_p$ and $v_p$ are all different vertices;
\item[(b)]  $u_p$ and $v_p$ denote the same vertex (i.e., $u_p=v_p$);
\item[(c)]  $u$ is the parent of $v$ (i.e., $u = v_p$), and $u$ is not the root vertex of $G$;
\item[(d)]  $u$ is the parent of $v$ (i.e., $u = v_p$), and $u$ is the root vertex of $G$.
\end{itemize}

In the analysis of the following cases, first we will consider
case (a). With the remaining cases (b), (c) and (d), we will proceed similarly.

\bigskip
\noindent {\it Case $1$}. $n_2=6$.

First we consider case (a), i.e., $u, v, u_p$ and
$v_p$ are pairwise distinct vertices.
Notice that in this case $d(v) = 7$ and $n_1=0$.


Here we apply the transformation $\mathcal{T}_{1}$ depicted in
Figure~\ref{noB4B2-20}.
\begin{figure}[h!]
\begin{center}
\includegraphics[scale=0.90]{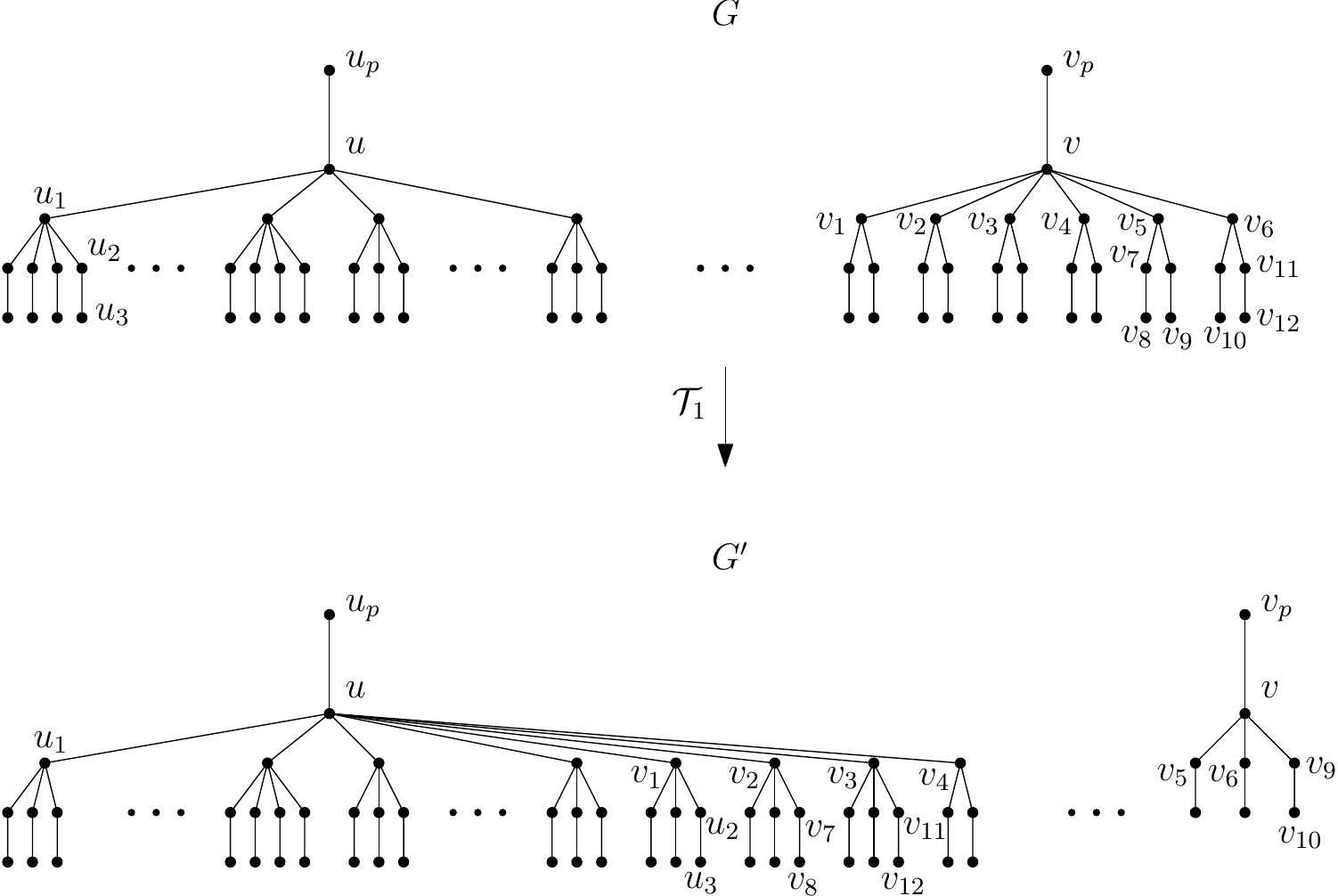}
\caption{An illustration of the transformation $\mathcal{T}_{1}$
from the proof of Theorem~\ref{thm-noB2B4} - Case $1$.}
\label{noB4B2-20}
\end{center}
\end{figure}

After applying $\mathcal{T}_{1}$, the degree of the vertex
$u$ increases by $4$, the degrees of the vertices $v_1, v_2,
v_3$ and $v_9$ increase by $1$, the degree of $v$ decreases by $3$,
while the degrees of $u_1, v_5$, $v_6$, one child of $v_5$ and one
child of $v_6$ decrease by $1$, and the rest of the vertices do not
change their degrees.

The change of the ABC index after applying $\mathcal{T}_{1}$ is
\beq \label{eq-te1-10}
ABC(G') - ABC(G) & = & \sum_{i=1}^{d(u)-1} (-f(d(x_i),d(u))+f(d(x_i),d(u)+4))\nonumber \\
& & -f(5,d(u))+f(4,d(u)+4)  \nonumber \\
&& +3(-f(3,7)+f(4,d(u)+4))   \nonumber  \\
&& -f(3,7)+f(3,d(u)+4)   \nonumber  \\
&& -f(d(v_p),7)+f(d(v_p),4)   \nonumber \\
&& +2(-f(3,7)+f(2,4)),
\eeq
where $x_i$, for $i=1, \dots, d(u)-1$, are the neighbors of $u$ in $G$,
except $u_1$.

Clearly,
$$
-f(d(x_i),d(u))+f(d(x_i),d(u)+4) < 0
$$
for all $i=1, \dots, d(u)-1$,
and both $f(4,d(u)+4)$ and $f(3,d(u)+4)$ decrease in $d(u)$. On the
other hand, by Lemmas \ref{appendix-pro-010-2} and
\ref{appendix-pro-010}, $-f(5,d(u))+f(4,d(u)+4)$ and
$-f(d(v_p),7)+f(d(v_p),4)$ increase in $d(u)$ and $d(v_p)$,
respectively, which implies that
$$
-f(5,d(u))+f(4,d(u)+4) < \lim_{d(u) \to \infty}
(-f(5,d(u))+f(4,d(u)+4)) = -\sqrt{\frac{1}{5}}+\sqrt{\frac{1}{4}}
$$
and
$$
-f(d(v_p),7)+f(d(v_p),4) < \lim_{d(v_p) \to \infty}
(-f(d(v_p),7)+f(d(v_p),4) ) =
-\sqrt{\frac{1}{7}}+\sqrt{\frac{1}{4}}.
$$

Now we have \beq \label{eq-te1-10-1} ABC(G') - ABC(G)
&< & -\sqrt{\frac{1}{5}}+\sqrt{\frac{1}{4}}  \nonumber  \\
&& +3(-f(3,7)+f(4,66+4))\nonumber  \\
&&-f(3,7)+f(3,66+4)  \nonumber  \\
&&  -\sqrt{\frac{1}{7}}+\sqrt{\frac{1}{4}}\nonumber  \\
&& +2(-f(3,7)+f(2,4)) \nonumber \\
&\approx & -0.0115077, \eeq since $d(u)>66$, when $d(v) = 7$.
Thus, we have shown that the change of the ABC index after applying
the transformation $\mathcal{T}_{1}$ is negative, which is a
contradiction to the initial assumption that $G$ is a minimal-ABC tree.

In the above deduction for case (a), notice that the term on
$d(u_p)$ may be neglected, thus, in case (b), i.e., when $u_p=v_p$, we
may obtain the same negative upper bound for the change of the ABC
index after applying $\mathcal{T}_{1}$ as in (\ref{eq-te1-10-1}).

For case (c), i.e., when $u = v_p$ and $u$ is not the root vertex of
$G$, we can obtain the same upper bound for the change of the ABC
index after applying $\mathcal{T}_{1}$, just by replacing $d(v_p)$
with $d(u)$ in (\ref{eq-te1-10}), i.e., \beq \label{eq-te1-10-2}
ABC(G') - ABC(G) & = & \sum_{i=1}^{d(u)-2} (-f(d(x_i),d(u))+f(d(x_i),d(u)+4))\nonumber \\
& & -f(5,d(u))+f(4,d(u)+4)  \nonumber \\
&& +3(-f(3,7)+f(4,d(u)+4))   \nonumber  \\
&& -f(3,7)+f(3,d(u)+4)   \nonumber  \\
&& -f(d(u),7)+f(d(u),4)   \nonumber \\
&& +2(-f(3,7)+f(2,4)). \eeq Furthermore, similar to the argument
regarding (\ref{eq-te1-10-1}), we can obtain the same negative upper
bound for $ABC(G') - ABC(G)$ as in (\ref{eq-te1-10-1}).

Notice that (\ref{eq-te1-10-2}) is independent of that if $u$ is the
root vertex of $G$ (i.e., the existence of the parent $u_p$ of $u$),
thus the upper bounds in cases (c) and (d) are actually the
same.

%
%

\bigskip
\noindent {\it Case $2$}. $n_2=5$.

First we consider case (a), i.e., $u, v, u_p$ and
$v_p$ are pairwise distinct vertices. In this case, either $d(v)=7$ (i.e., $n_1=1$) or $d(v)=6$ (i.e., $n_1=0$).

Here, we apply the transformation $\mathcal{T}_{2}$ depicted in
Figure~\ref{noB4B2-30}.
\begin{figure}[h!]
\begin{center}
\includegraphics[scale=0.90]{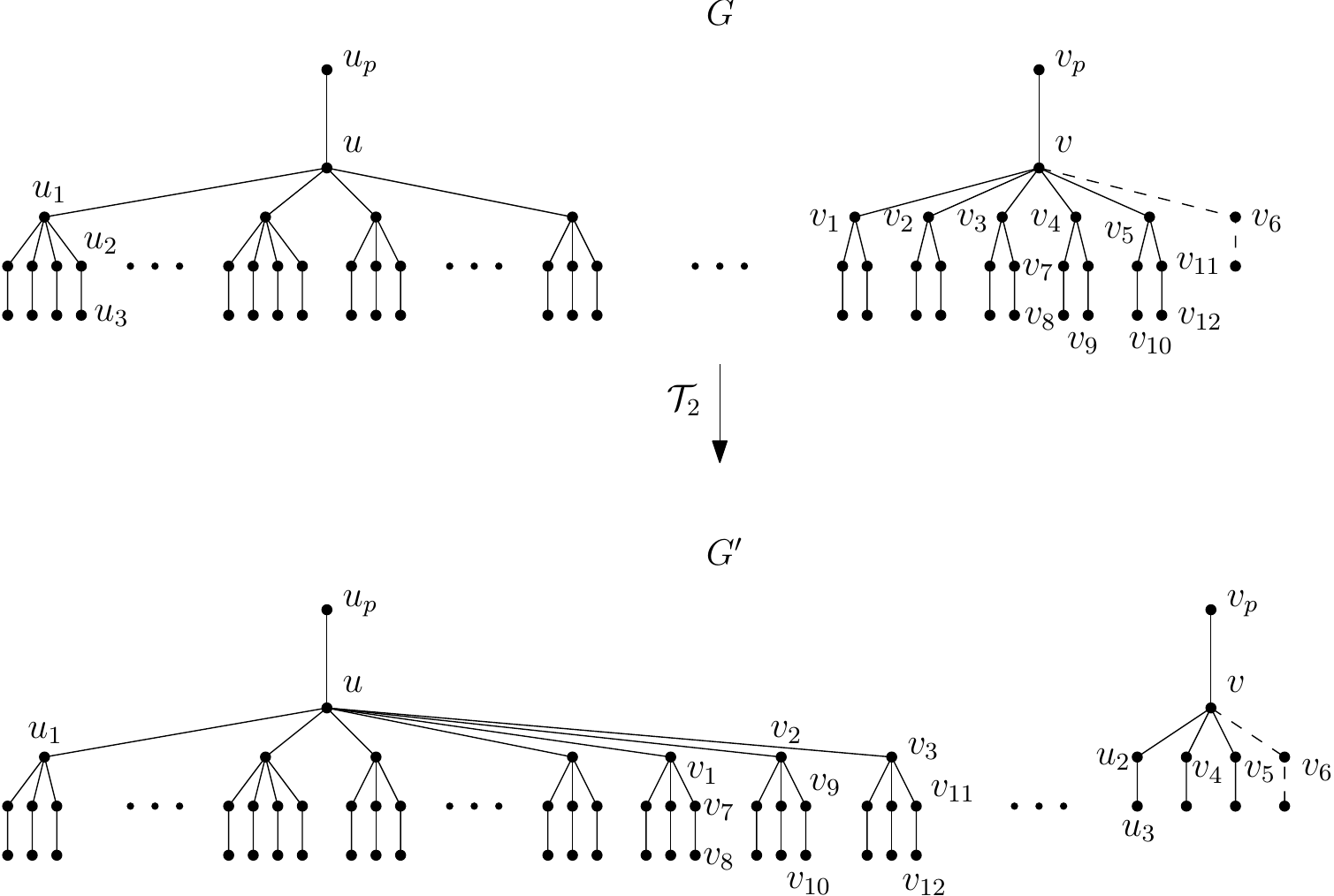}
\caption{An illustration of the transformation $\mathcal{T}_{2}$
from the proof of Theorem~\ref{thm-noB2B4} - Case $2$.  The dashed
lines indicate that the attached branches may be optional.} \label{noB4B2-30}
\end{center}
\end{figure}
After applying $\mathcal{T}_{2}$, the degree of the vertex $u$
increases by $3$, the degrees of the vertices $v_1, v_2, v_3$ and
$v_{9}$ increase by $1$, the degree of $v$ decreases by $2$, while
the degrees of $u_1, v_4$, $v_5$, one child of $v_4$ and one child
of $v_5$ decrease by $1$, and the rest of the vertices do not change
their degrees.

The change of the ABC index after applying $\mathcal{T}_{2}$ is
\beq \label{eq-te1-40}
ABC(G') - ABC(G) & = & \sum_{i=1}^{d(u)-1} (-f(d(x_i),d(u))+f(d(x_i),d(u)+3))  \nonumber \\
& & -f(5,d(u))+f(4,d(u)+3)  \nonumber \\
&& +3(-f(3,d(v))+f(4,d(u)+3))   \nonumber  \\
&& -f(d(v_p),d(v))+f(d(v_p),d(v)-2)   \nonumber \\
&& +2(-f(3,d(v))+f(2,4)), \nonumber \eeq
where $x_i$, for $i=1,\dots,d(u)-1$, are all the neighbors of $u$ in
$G$, except $u_1$.

A similar analysis as in Case~$1$ shows that \beq
\label{eq-te1-40-1}
ABC(G') - ABC(G) & < &  -\sqrt{\frac{1}{5}} + \sqrt{\frac{1}{4}}  \nonumber  \\
&& +3(-f(3,d(v))+f(4,d(u)+3))  \nonumber  \\
&& -\sqrt{\frac{1}{d(v)}} + \sqrt{\frac{1}{d(v)-2}}\nonumber  \\
&&+2(-f(3,d(v))+f(2,4)).
\eeq

Recall that here $d(v)=6$ or $d(v)=7$, and then $d(u) \geq 25$ or  $d(u) \geq 67$, respectively.
Observe that the right-hand side of (\ref{eq-te1-40-1}) decreases in $d(u)$.
Thus, for $d(v) = 6$, we obtain the upper bound on (\ref{eq-te1-40-1}) when $d(u) =25$,
which is $-0.00664864$.
In the case of $d(v)=7$, the upper bound on (\ref{eq-te1-40-1}) is obtained when $d(u) =67$
and it is  $-0.0285403$. 


The proofs for cases (b), (c) and (d) are similar to case (a), and
the detailed illustration can be referred to that in Case $1$.

\bigskip
\noindent {\it Case $3$}. $n_2=4$.

First we consider case (a), i.e., $u, v, u_p$ and
$v_p$ are all different vertices.
In this case, $d(v)=7, 6, 5$ and $n_1=2,1,0$, correspondingly.

Here, we apply the transformation $\mathcal{T}_{3}$ depicted in
Figure~\ref{noB4B2-40}.
\begin{figure}[h!]
\begin{center}
\includegraphics[scale=0.90]{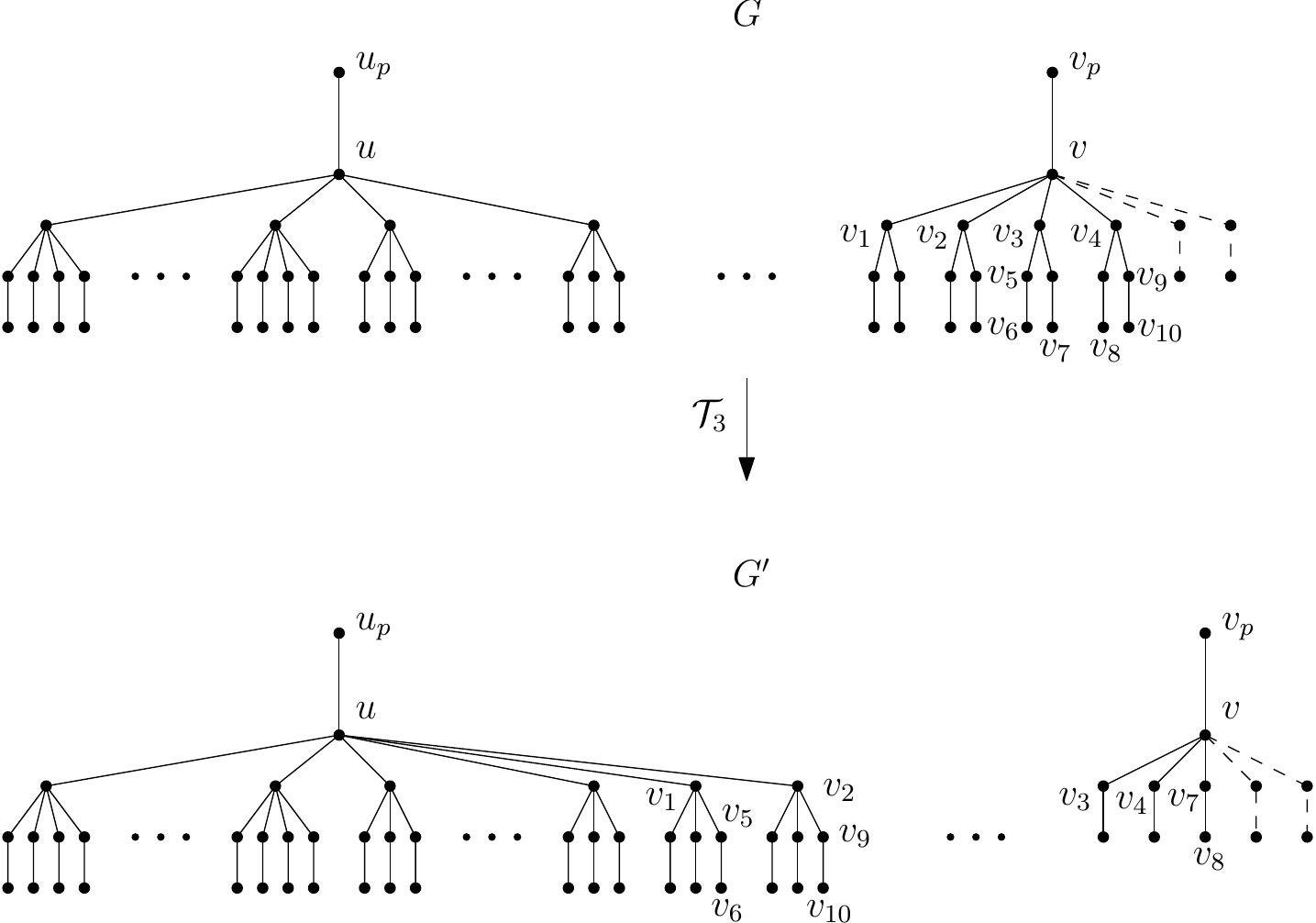}
\caption{An illustration of the transformation $\mathcal{T}_{3}$
from the proof of Theorem~\ref{thm-noB2B4} - Case $3$. The dashed
lines indicate that the attached branches may be optional.}
\label{noB4B2-40}
\end{center}
\end{figure}
After applying $\mathcal{T}_{3}$, the degree of the vertex $u$
increases by $2$, the degrees of the vertices $v_1, v_2$ and
$v_{7}$ increase by $1$, the degrees of $v, v_3, v_4$, one child of
$v_3$ and one child of $v_4$ decrease by $1$, and the rest of the
vertices do not change their degrees.

The change of the ABC index after applying $\mathcal{T}_{3}$ is
\beq \label{eq-te1-50}
ABC(G') - ABC(G) & = &\sum_{i=1}^{d(u)} (-f(d(x_i),d(u))+f(d(x_i),d(u)+2)) \nonumber \\
&& +2(-f(3,d(v))+f(4,d(u)+2))   \nonumber  \\
&& -f(d(v_p),d(v))+f(d(v_p),d(v)-1)   \nonumber \\
&& +2(-f(3,d(v))+f(2,4)),
\eeq
where $x_i$, for $i=1, \dots, d(u)$, are all the neighbors of $u$ in $G$.

Similar analysis as in Case~$1$ shows that \beq \label{eq-te1-50-1}
ABC(G') - ABC(G) & < & 2(-f(3,d(v))+f(4,d(u)+2)) \nonumber  \\
&& -\sqrt{\frac{1}{d(v)}} + \sqrt{\frac{1}{d(v)-1}}\nonumber  \\
&& +2(-f(3,d(v))+f(2,4)).
\eeq
Clearly, the right-hand side of
(\ref{eq-te1-50-1}) decreases in $d(u)$, so the negative change of
the ABC index follows again from direct calculation, except the
cases  $d (v) = 5$ and $d (u) = 13$.
For this case, we need only to analyze the following term in
(\ref{eq-te1-50}):
$$
\sum_{i=1}^{d(u)} (-f(d(x_i),d(u))+f(d(x_i),d(u)+2)).
$$
Note that, from Lemma \ref{appendix-pro-010-2},
$-f(d(x_i),d(u))+f(d(x_i),d(u)+2)$ decreases in $d(x_i)$, and from
Lemmas \ref{lemma-switching} and \ref{lemma-10}(b), the possible minimum degree
among all the neighbors of $u$ in $G$, different from $u_1$ and $u_p$, is
$4$. Thus
$$
\sum_{i=1}^{d(u)} (-f(d(x_i),d(u))+f(d(x_i),d(u)+2)) < (d(u)-2)
(-f(4,d(u))+f(4,d(u)+2)).
$$
Now, it follows that \beq
ABC(G') - ABC(G) & < & (d(u)-2) (-f(4,d(u))+f(4,d(u)+2)) \nonumber \\
&& +2(-f(3,d(v))+f(4,d(u)+2))   \nonumber  \\
&& -\sqrt{\frac{1}{d(v)}} + \sqrt{\frac{1}{d(v)-1}}   \nonumber \\
&& +2(-f(3,d(v))+f(2,4)). \nonumber \eeq
Subsequently a negative upper bound
follows from direct calculation, for $d (v) = 5$ and $d (u) = 13$.

The proofs for cases (b), (c) and (d) are similar to that in case
(a), and the detailed illustration can be referred to that in Case $1$.

\bigskip
\noindent {\it Case $4$}. $n_2=3$.

First we consider  case (a), i.e., $u, v, u_p$ and
$v_p$ are pairwise distinct vertices.
In this case, $d(v)=7, 6, 5$ and $n_1=3,2,1$, correspondingly.
We distinguish two subcases regarding the degree of
the vertex $v$.

\bigskip
\noindent {\it Subcase $4.1$}. $d(v)=5,6$.

\noindent Here, we apply the transformation $\mathcal{T}_{41}$
depicted in Figure~\ref{noB4B2-50-1}.
\begin{figure}[h!]
\begin{center}
\includegraphics[scale=0.90]{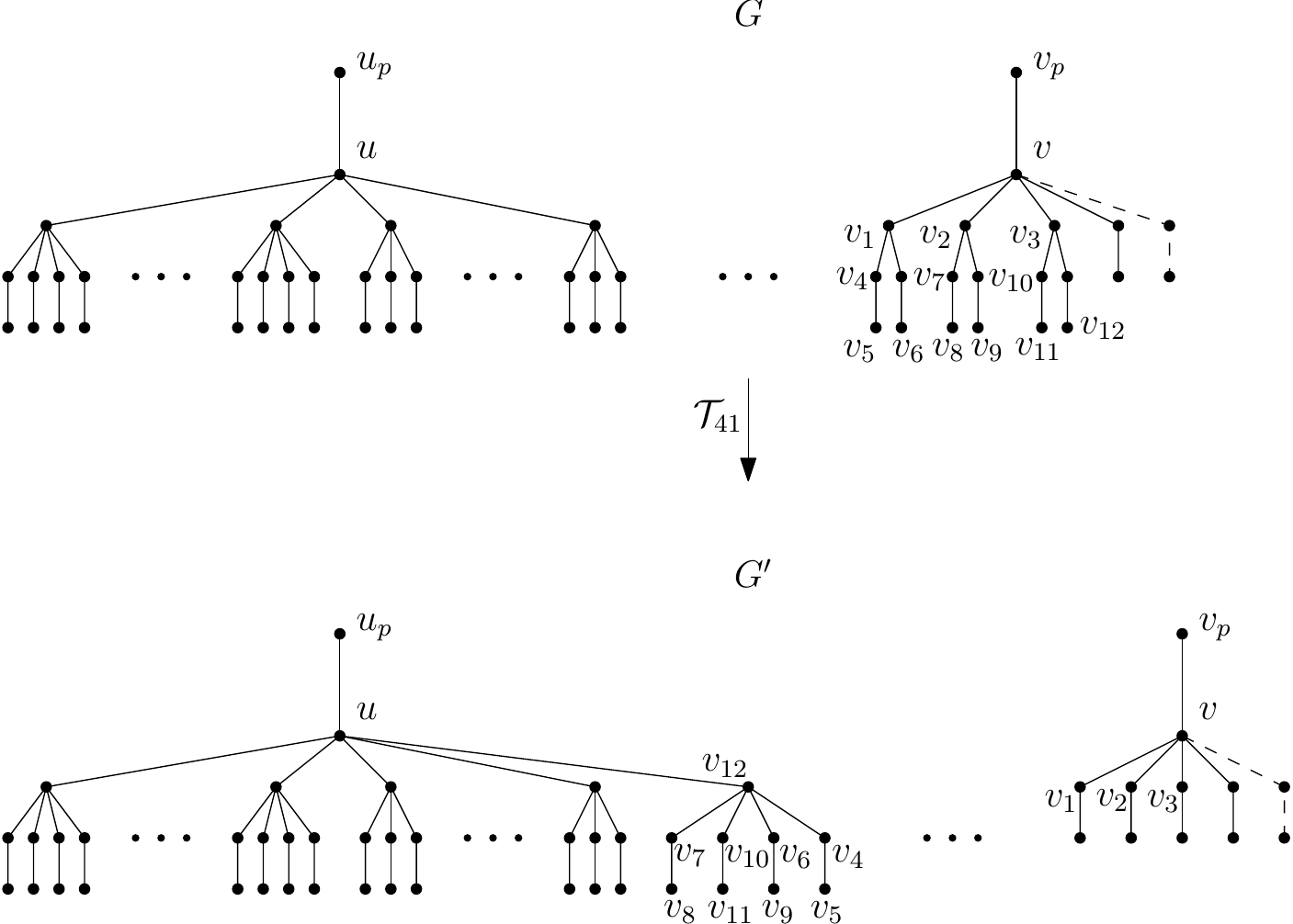}
\caption{Transformation $\mathcal{T}_{41}$  from the proof of Theorem~\ref{thm-noB2B4} - Subcase $4.1$,
which results in a negative change of the ABC index  for $d(v)=5,6$.
}
\label{noB4B2-50-1}
\end{center}
\end{figure}
After applying $\mathcal{T}_{41}$, the degrees of the vertices $u$
and $v_6$ increase by $1$, the vertex $v_{12}$ increases its degree
from $1$ to $5$, while $v_1, v_2, v_3$ and one child from each decrease their degrees by $1$, the rest of the vertices do not
change their degrees.

The change of the ABC index after applying $\mathcal{T}_{41}$ is
\beq \label{eq-te1-60}
ABC(G') - ABC(G)
&= & \sum_{i=1}^{d(u)} (-f(d(x_i),d(u))+f(d(x_i),d(u)+1))  \nonumber \\
&& -f(3,d(v))+f(5,d(u)+1) \nonumber \\
&&+2(-f(3,d(v))+f(2,d(v))),
\eeq
where $x_i$, for $i=1, \dots, d(u)$, are all the neighbors of $u$ in $G$.


Clearly,
$$
-f(d(x_i),d(u))+f(d(x_i),d(u)+1) < 0
$$
for all $i=1, \dots, d(u)$, thus
one can obtain that
\beq \label{eq-te1-60-1}
ABC(G') - ABC(G) & < & -f(3,d(v))+f(5,d(u)+1) \nonumber \\
&&+2(-f(3,d(v))+f(2,d(v))).
\eeq
Notice that the right-hand side of
(\ref{eq-te1-60-1}) decreases in $d(u)$, thus the negative change of the ABC
index follows, except the cases where
\begin{itemize}
\item $d (v) = 5$ and $13 \le d (u) \le 16$;
\item $d (v) = 6$ and $25 \le d (u) \le 69$.
\end{itemize}

For the remaining cases, we need only to consider
(\ref{eq-te1-60}), in particular the term
$$
\sum_{i=1}^{d(u)} (-f(d(x_i),d(u))+f(d(x_i),d(u)+1)).
$$
Note that, from Lemma \ref{appendix-pro-010-2},
$-f(d(x_i),d(u))+f(d(x_i),d(u)+1)$ decreases in $d(x_i)$, and from
Lemmas \ref{lemma-switching} and \ref{lemma-10}(b), the possible minimum degree
among all the neighbors of $u$ in $G$, different from $u_1$ and $u_p$, is
$4$. Hence
$$
\sum_{i=1}^{d(u)} (-f(d(x_i),d(u))+f(d(x_i),d(u)+1)) < (d(u)-2)
(-f(4,d(u))+f(4,d(u)+1)).
$$
Now, it follows that
\beq
ABC(G') - ABC(G) & < & (d(u)-2) (-f(4,d(u))+f(4,d(u)+1)) \nonumber \\
&& -f(3,d(v))+f(5,d(u)+1)\nonumber \\
&& +2(-f(3,d(v))+f(2,d(v))).    \nonumber
\eeq
Consequently we get a negative
upper bound by direct calculation.

\bigskip
\noindent {\it Subcase $4.2$}. $d(v)=7$.

In this case, the number of the $B_1$-branches is $3$, i.e., $n_1 = 3$.
%
Here, we apply
the transformation $\mathcal{T}_{42}$ depicted in
Figure~\ref{noB4B2-50-2}.
\begin{figure}[h!]
\begin{center}
\includegraphics[scale=0.90]{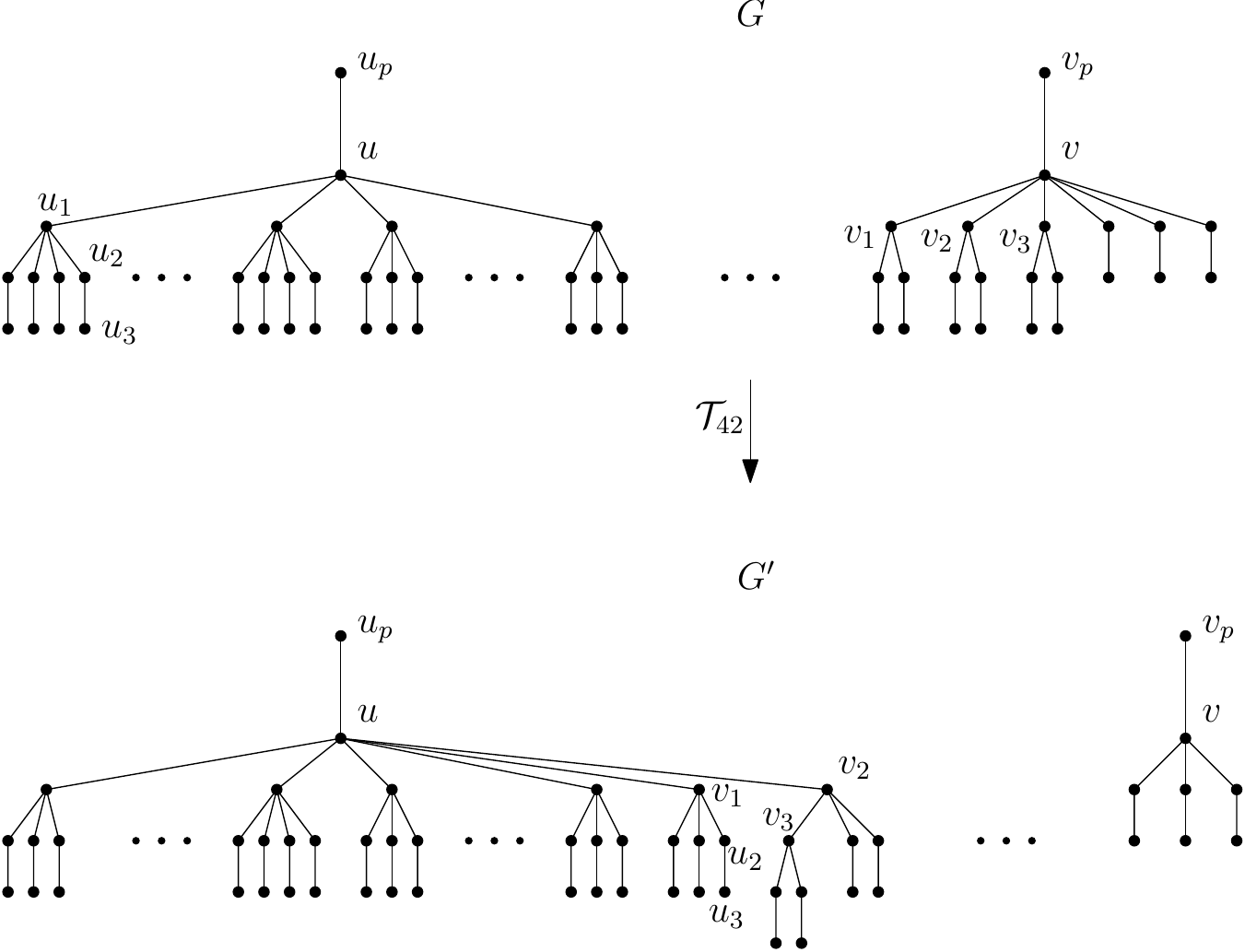}
\caption{Transformation $\mathcal{T}_{42}$ from the proof of
Theorem~\ref{thm-noB2B4} -  Subcase $4.2$ which results in a
negative change of the ABC index  for $d(v)=7$.} \label{noB4B2-50-2}
\end{center}
\end{figure}
%
%
After applying $\mathcal{T}_{42}$, the degree of the vertex $u$
increases by $2$, the degrees of $v_1$ and $v_2$ increase from $3$ to
$4$, the degree of $v$ after applying $\mathcal{T}_{42}$ is $4$, the degree of $u_1$  decreases from $5$ to $4$, and the rest of the
vertices do not change their degrees.

The change of the ABC index after applying $\mathcal{T}_{42}$ is
\beq \label{eq-te1-70}
ABC(G') - ABC(G) & =& \sum_{i=1}^{d(u)-1} (-f(d(x_i),d(u))+f(d(x_i),d(u)+2))  \nonumber \\
&& -f(5,d(u))+f(4,d(u)+2) \nonumber \\
&&+2(-f(3,7)+f(4,d(u)+2))   \nonumber  \\
&&-f (d(v_p),7)+f(d(v_p),4)  \nonumber \\
&&-f(3,7)+f(4,3), \eeq
where $x_i$, for $i=1, \dots, d(u)-1$, are the neighbors of $u$ in $G$,
except $u_1$.
%
An analogous analysis to Case~$1$ shows that
\beq \label{eq-te1-70-1}
ABC(G') - ABC(G) & < & -\sqrt{\frac{1}{5}} + \sqrt{\frac{1}{4}}\nonumber \\
&&+2(-f(3,7)+f(4,d(u)+2))   \nonumber  \\
&& -\sqrt{\frac{1}{7}} + \sqrt{\frac{1}{4}} \nonumber \\
&& -f(3,7)+f(4,3).
\eeq
Clearly, the
right-hand side of (\ref{eq-te1-70-1}) decreases in $d(u)$, so the
negative change of the ABC index follows.

%
%

The proofs for cases (b), (c) and (d) are similar to case (a), and
the detailed illustration can be referred to that in Case $1$.

\bigskip
\noindent {\it Case $5$}. $n_2=2$.

First we consider  case (a).
In this case, since $d(v)=7, 6, 5$, we have that $n_1=4,3,2$, correspondingly.
%
%
%
\noindent Here, we apply the transformation
$\mathcal{T}_{5}$ depicted in Figure~\ref{noB4B2-60-1}.
\begin{figure}[h!]
\begin{center}
\includegraphics[scale=0.90]{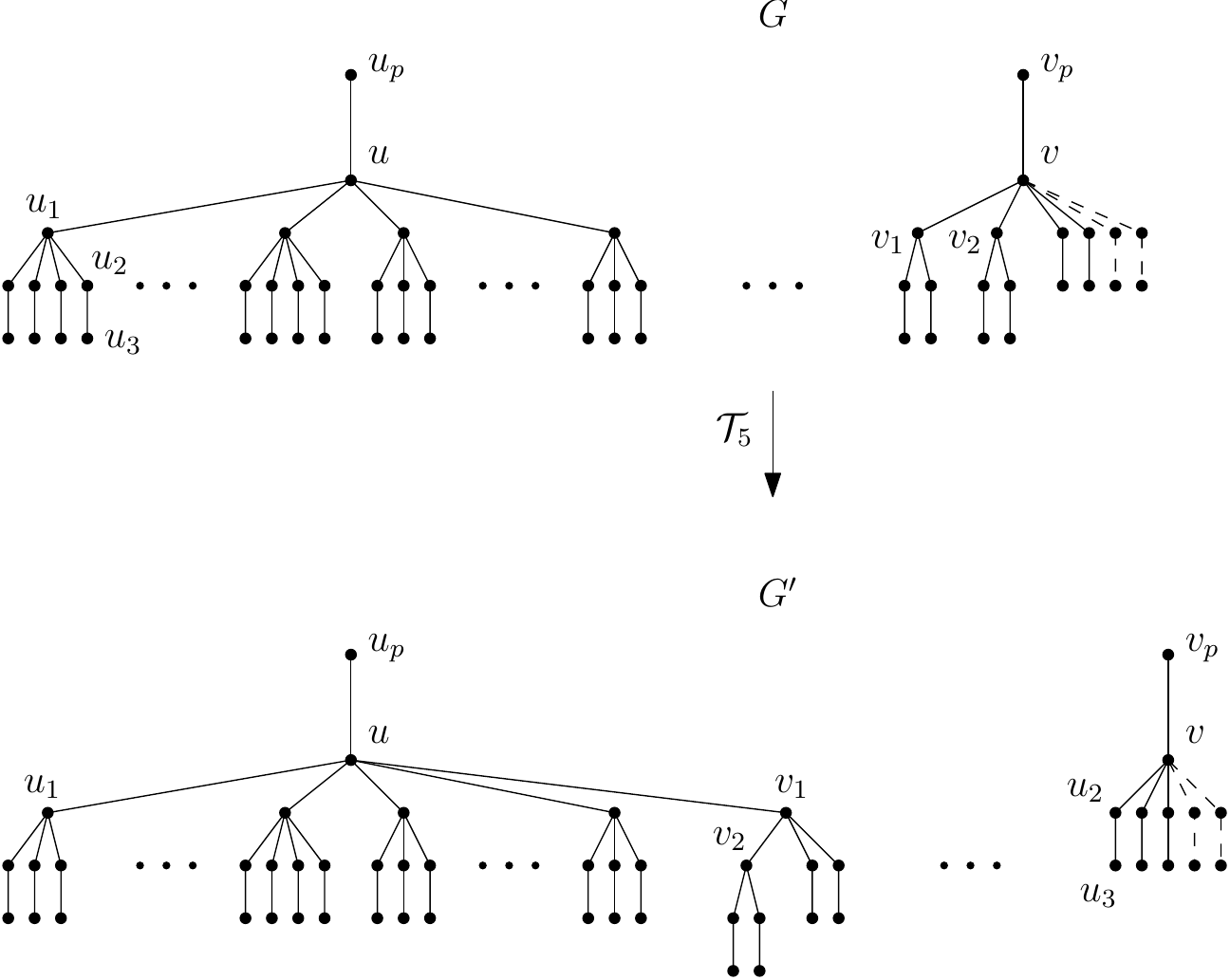}
\caption{Transformation $\mathcal{T}_{5}$ from the proof of Theorem~\ref{thm-noB2B4} -  Case $5$.
}
\label{noB4B2-60-1}
\end{center}
\end{figure}
After applying $\mathcal{T}_{5}$, the degree of the vertex $u$ increases by $1$, the degree of $v_1$ increases from $3$ to
$4$, the degree of the vertex $v$ decreases by $1$, the degree of  $u_1$  decreases from $5$ to $4$, and the
rest of the vertices do not change their degrees.

The change of the ABC index after applying $\mathcal{T}_{5}$ is
\beq \label{eq-te1-80}
ABC(G') - ABC(G) & = & \sum_{i=1}^{d(u)-1} (-f(d(x_i),d(u))+f(d(x_i),d(u)+1))  \nonumber \\
&& -f(5,d(u))+f(4,d(u)+1)  \nonumber  \\
&& -f(3,d(v))+f(4,3)   \nonumber  \\
&&-f(3,d(v))+f(4,d(u)+1) \nonumber  \\
&& -f(d(v_p),d(v))+f(d(v_p),d(v)-1), \eeq
where $x_i$, for $i=1, \dots, d(u)-1$, are the neighbors of $u$ in $G$,
except $u_1$.
Similar technique in Case~$1$ shows that
\beq \label{eq-te1-80-1}
ABC(G') - ABC(G) & < & -\sqrt{\frac{1}{5}} + \sqrt{\frac{1}{4}}   \nonumber  \\
&& -f(3,d(v))+f(4,3) \nonumber  \\
&&-f(3,d(v))+f(4,d(u)+1) \nonumber  \\
&& -\sqrt{\frac{1}{d(v)}}+\sqrt{\frac{1}{d(v)-1}}.
\eeq
Clearly, the
right-hand side of (\ref{eq-te1-80-1}) decreases in $d(u)$, so the
fact that the change of the ABC index being negative follows from
direct calculation, except the cases where
\begin{itemize}
\item
$d(v) = 5$ and $13 \le d (u) \le 34$;
\item
$d(v) = 6$ and  $25 \le d (u) \le 48$;
\item
$d(v) = 7$ and  $67 \le d (u) \le 83$.
\end{itemize}

For the above cases, we need only to analyze in
(\ref{eq-te1-80}) the term
$$
\sum_{i=1}^{d(u)-1} (-f(d(x_i),d(u))+f(d(x_i),d(u)+1)).
$$
Note that, from Lemma \ref{appendix-pro-010-2},
$-f(d(x_i),d(u))+f(d(x_i),d(u)+1)$ decreases in $d(x_i)$,
and from Lemmas \ref{lemma-switching} and \ref{lemma-10}(b), the possible
minimum degree among all the neighbors of $u$ in $G$, different from $u_1$
and $u_p$, is $4$. Therefore, we have
\beq \label{eq-te1-80-50}
&& \sum_{i=1}^{d(u)-1} (-f(d(x_i),d(u))+f(d(x_i),d(u)+1))  \nonumber  \\
&<&  \quad  (d(u)-2) (-f(4,d(u))+f(4,d(u)+1)).   \nonumber
\eeq
Now, we obtain
\beq
ABC(G') - ABC(G) & < & (d(u)-2) (-f(4,d(u))+f(4,d(u)+1)) \nonumber  \\
&&-\sqrt{\frac{1}{5}} + \sqrt{\frac{1}{4}}\nonumber \\
&&  -f(3,d(v))+f(4,3)  \nonumber  \\
&&-f(3,d(v))+f(4,d(u)+1)  \nonumber  \\
&& -\sqrt{\frac{1}{d(v)}}+\sqrt{\frac{1}{d(v)-1}}.   \nonumber \eeq
Subsequently a negative upper bound follows from direct calculation.


The proofs for cases (b), (c) and (d) are similar to that in case
(a), and the detailed illustration can be referred to that in Case $1$.

\bigskip
\noindent {\it Case $6$}. $n_2=1$.

First we consider case (a), i.e., $u, v, u_p$ and
$v_p$ are pairwise distinct vertices.
By Theorem~\ref{thm-30}  $n_1 \leq 4$. Thus,  there are two possible configurations in this case:   $d(v)=6, 5$ and  $n_1=4,3$, correspondingly.
Here, we apply the transformation
$\mathcal{T}_{6}$ depicted in Figure~\ref{noB4B2-70-2}.
\begin{figure}[h!]
\begin{center}
\includegraphics[scale=0.90]{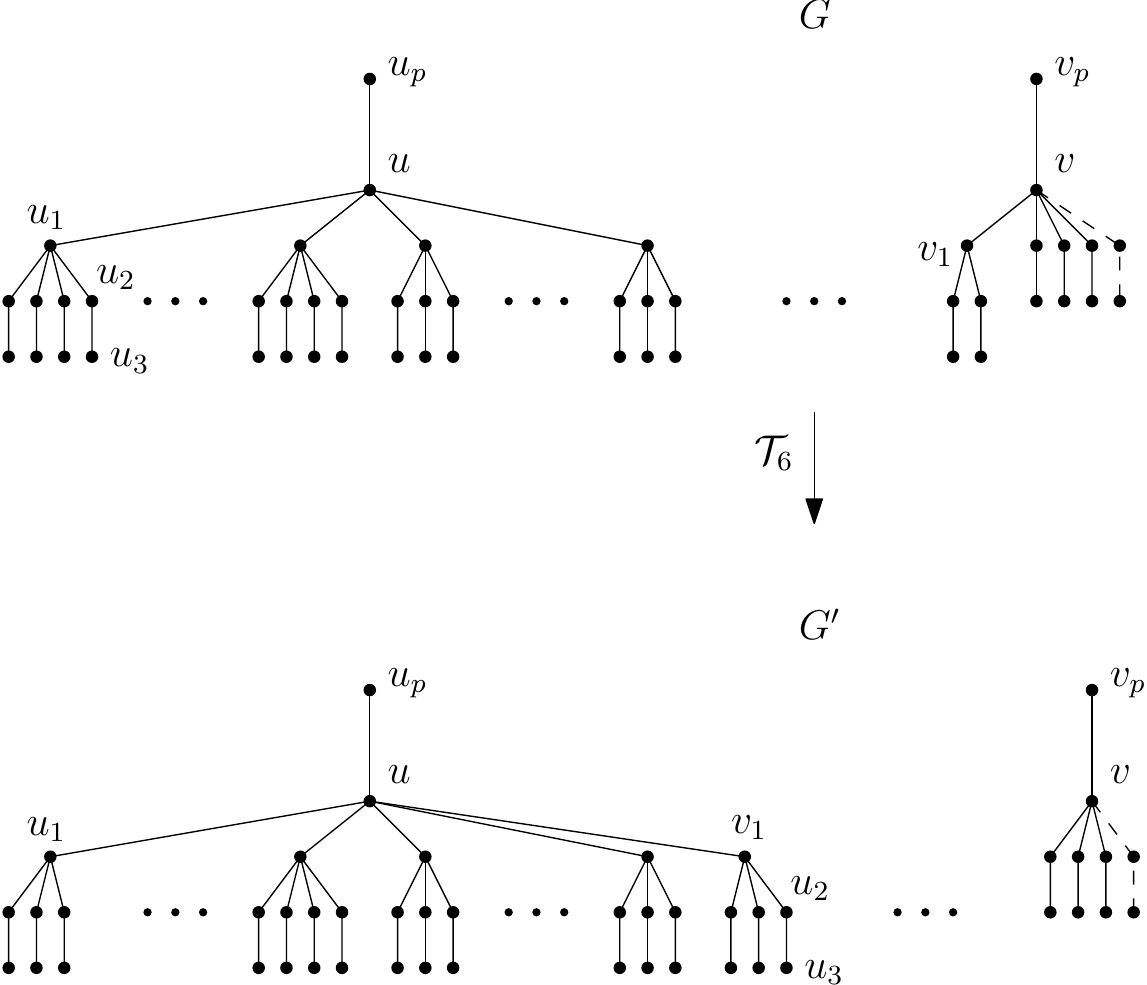}
\caption{Transformation $\mathcal{T}_{6}$ from the proof of Theorem~\ref{thm-noB2B4} -  Case $6$.
}
\label{noB4B2-70-2}
\end{center}
\end{figure}
After applying $\mathcal{T}_{6}$, the degree of the vertex $u$
increases by $1$, the degree of $v_1$ increases from $3$ to
$4$, the degree of the vertex $v$ decreases by $1$, the degree of  $u_1$  decreases from $5$ to $4$, and the
rest of the vertices do not change their degrees.

The change of the ABC index after applying $\mathcal{T}_{6}$ is
\beq \label{eq-te1-100}
ABC(G') - ABC(G) & = & \sum_{i=1}^{d(u)-1} (-f(d(x_i),d(u))+f(d(x_i),d(u)+1))  \nonumber \\
&& -f(5,d(u))+f(4,d(u)+1)  \nonumber  \\
&&-f(3,d(v))+f(4,d(u)+1) \nonumber  \\
&& -f(d(v_p),d(v))+f(d(v_p),d(v)-1), \eeq
where $x_i$, for $i=1, \dots, d(u)-1$, are the neighbors of $u$ in $G$
different from $u_1$.
A similar analysis as in Case~$1$ shows that
\beq \label{eq-te1-100-1}
ABC(G') - ABC(G) &<& -\sqrt{\frac{1}{5}}+\sqrt{\frac{1}{4}}   \nonumber  \\
&& -f(3,d(v))+f(4,d(u)+1)  \nonumber  \\
&& -\sqrt{\frac{1}{d(v)}}+\sqrt{\frac{1}{d(v)-1}}.
\eeq
Clearly, the
right-hand side of (\ref{eq-te1-100-1}) decreases in $d(u)$, so we
can get a negative upper bound through direct calculation, except for the
case when $d (v) = 5$ and $13 \le d(u) \le 17$.

For the remaining cases, we need only to consider the term
(\ref{eq-te1-100}), which is as follows
$$
\sum_{i=1}^{d(u)-1} (-f(d(x_i),d(u))+f(d(x_i),d(u)+1)).
$$
Note that, from Lemma \ref{appendix-pro-010-2},
$-f(d(x_i),d(u))+f(d(x_i),d(u)+1)$ decreases in $d(x_i)$,
and from Lemmas \ref{lemma-switching} and \ref{lemma-10}(b), the possible
minimum degree among all the neighbors of $u$ in $G$, different from $u_1$
and $u_p$, is $4$, thus
\beq
&&\sum_{i=1}^{d(u)-1} (-f(d(x_i),d(u))+f(d(x_i),d(u)+1)) \nonumber \\
&<& (d(u)-2) (-f(4,d(u))+f(4,d(u)+1)). \nonumber
\eeq
Now what follows is that
\beq
ABC(G') - ABC(G) &<& (d(u)-2) (-f(4,d(u))+f(4,d(u)+1)) \nonumber  \\
&&-\sqrt{\frac{1}{5}}+\sqrt{\frac{1}{4}} \nonumber  \\
&& -f(3,d(v))+f(4,d(u)+1) \nonumber  \\
&& -\sqrt{\frac{1}{d(v)}}+\sqrt{\frac{1}{d(v)-1}}.   \nonumber \eeq
Subsequently a negative upper bound follows from direct calculation.

The proofs for cases (b), (c) and (d) are similar to that in case
(a), and as before the detailed illustration can be referred to that in Case $1$.

By combining the above six cases, the claim of the theorem is finally obtained.
\end{proof}

\section{Trees containing simultaneously $B_4$- and  $B_1$-branches}\label{sec-B4B1}

In the final section we prove that a minimal-ABC tree does not contain a
$B_4$-branch and a $B_1$-branch in simultaneity.

\begin{te} \label{thm-noB1B4}
A minimal-ABC tree cannot contain a $B_4$-branch and a $B_1$-branch
simultaneously.
\end{te}
\begin{proof}

\noindent Let $u_1$ be the center vertex of a $B_4$-branch and $v_1$
the center vertex of a $B_1$-branch in a minimal-ABC tree $G$. By
Lemma~\ref{lemma-10}(a), $u_1$ and $v_1$ cannot have a common parent
vertex.
So, let $u$ be the parent of $u_1$, and $u_p$ the parent of $u$ if
$u$ is not the root vertex of $G$. Similarly, let $v$ be the parent
of $v_1$, and $v_p$ the parent of $v$.

On one hand, by Theorem~\ref{thm-noB2B4} and
Lemma~\ref{lemma-10}(a), neither $B_2$- nor $B_4$-branch can be
attached to $v$. On the other hand, since $v$ has $B_1$-branches as children, thus $v$ has a child of degree at least $3$, i.e.,
there must exist some $B_3$-branches attached to $v$, which also
implies that no $B_1^*$-branch can be attached to $v$ from Lemma
\ref{lemma-10-1}. In conclusion, the branches attached to $v$ can
only be $B_1$- or $B_3$-branches.
Here  also, by Observation \ref{obs-DS-branches1}, we may assume that $u_1$ occurs before $v_1$ in the breadth-first search of $G$.
Furthermore, by Lemma~\ref{lemma-switching},
it follows that $d(u)\geq d(v) \geq d(u_1) \geq d(v_1)$.

First, we consider case (a), i.e., when the vertices $u, v, u_p$
and $v_p$ are pairwise distinct. Denote by $n_1$ the number of
$B_1$-branches attached to $v$ in $G$. Let us consider the
transformation $\mathcal{T}_{7}$ depicted in Figure~\ref{noB4B1-10}.
\begin{figure}[h!]
\begin{center}
\includegraphics[scale=0.90]{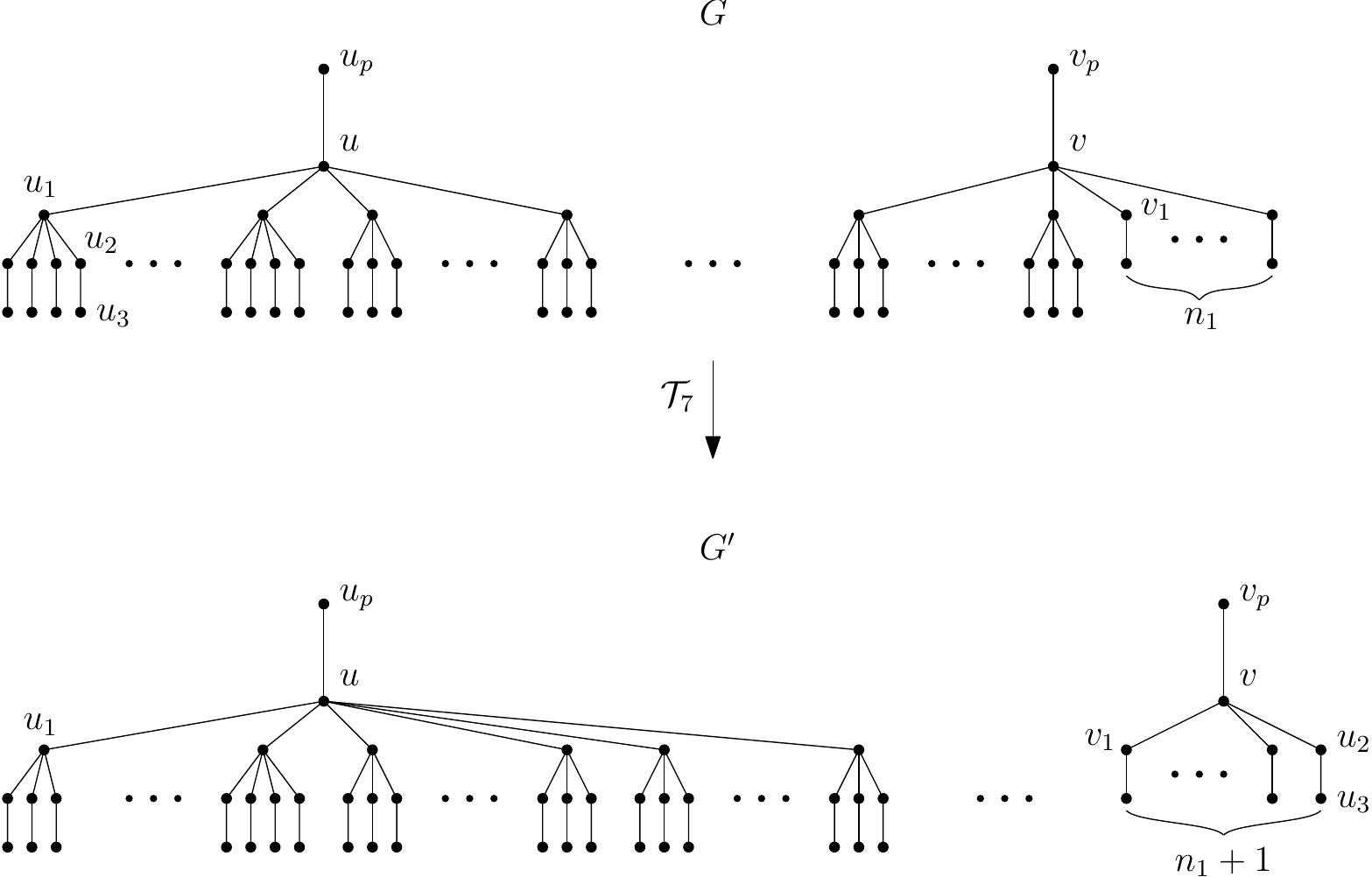}
\caption{Transformation $\mathcal{T}_{7}$ from the proof of
Theorem~\ref{thm-noB1B4}.} \label{noB4B1-10}
\end{center}
\end{figure}
After applying $\mathcal{T}_{7}$, the degree of $u$ increases by
$d(v)-n_1-1$, the degree of $v$ decreases by $d(v)-n_1-2$, while the
degree of  $u_1$  decreases by $1$, and the rest of the vertices do
not change their degrees.

The change of the ABC index after applying
$\mathcal{T}_{7}$ is
\beq \label{eq-te1-110}
ABC(G') - ABC(G) &=& \sum_{i=1}^{d(u)-1} (-f(d(x_i),d(u))+f(d(x_i),d(u)+d(v)-n_1-1)) \nonumber \\
&&-f(5,d(u))+f(4,d(u)+d(v)-n_1-1)  \nonumber  \\
&& +(d(v)-n_1-1)(-f(4,d(v))+f(4,d(u)+d(v)-n_1-1))  \nonumber  \\
&& -f(d(v_p),d(v))+f(d(v_p),n_1+2),  \nonumber \eeq
where $x_i$, for $i=1, \dots, d(u)-1$, are all the neighbors of $u$ in
$G$, with the exception of $u_1$.

We remark that $d(v) \ge n_1+2 > n_1 + 1$, since $v$ has
$B_1$-branches. Then, as previously, we obtain
\beq
\label{eq-te1-110-1} ABC(G') -
ABC(G) &<&
-\sqrt{\frac{1}{5}}+\sqrt{\frac{1}{4}}\nonumber \\
&& +(d(v)-n_1-1)(-f(4,d(v))+f(4,d(u)+d(v)-n_1-1))\nonumber \\
&& -\sqrt{\frac{1}{d(v)}}+\sqrt{\frac{1}{n_1+2}}.
\eeq
Clearly, the
right-hand side of (\ref{eq-te1-110-1}) decreases in $d(u)$, so we
can obtain a negative upper bound through direct calculation procedure.
Finally, the proofs for the remaining cases of (b), (c), and (d) are
very similar to this one.
\end{proof}

\section[Acknowledgment]{Acknowledgment}
The authors would like to thank the anonymous reviewers for their insightful comments.
The authors are especially  indebted to the reviewer
whose observations also aided in the shortening the proof of Theorem~\ref{thm-noB2B4}.

%
%


\begin{thebibliography}{999}
\setlength{\itemsep}{0pt}


\bibitem{adgh-dctmabci-14}
M.~B.~Ahmadi, D.~Dimitrov, I.~Gutman, S.~A.~Hosseini,
\textit{Disproving a conjecture on trees with minimal atom-bond connectivity index},
MATCH Commun. Math. Comput. Chem. \textbf{72} (2014) 685--698.


\bibitem{ahs-tmabci-13}
M.~B.~Ahmadi, S.~A.~Hosseini, P.~Salehi~Nowbandegani,
\textit{On trees with minimal atom bond connectivity index},
MATCH Commun. Math. Comput. Chem. \textbf{69} (2013) 559--563.


\bibitem{ahz-ltmabci-13}
M.~B.~Ahmadi, S.~A.~Hosseini, M.~Zarrinderakht,
\textit{On large trees with minimal atom--bond connectivity index},
MATCH Commun. Math. Comput. Chem. \textbf{69} (2013) 565--569.









\bibitem{cg-eabcig-11}
J.~Chen, X.~Guo,
\textit{Extreme atom-bond connectivity index of graphs},
MATCH Commun. Math. Comput. Chem. \textbf{65} (2011) 713--722.



\bibitem{clg-subabcig-12}
J.~Chen, J.~Liu, X.~Guo,
\textit{Some upper bounds for the atom-bond connectivity index of graphs},
Appl. Math. Lett. \textbf{25} (2012) 1077--1081.

\bibitem{cll-abcbsp-13}
J.~Chen, J.~Liu, Q.~Li,
\textit{The atom-bond connectivity index of catacondensed polyomino graphs},
Discrete Dyn. Nat. Soc. \textbf{2013}  (2013) ID 598517.


\bibitem{d-abcig-10}
K.~C.~Das,
\textit{Atom-bond connectivity index of graphs},
Discrete Appl. Math. \textbf{158} (2010) 1181--1188.

\bibitem{dgf-abci-11}
K.~C.~Das, I.~Gutman, B.~Furtula,
\textit{On atom-bond connectivity index},
Chem. Phys. Lett. \textbf{511} (2011) 452--454.

\bibitem{dgf-abci-12}
K.~C.~Das, I.~Gutman, B.~Furtula,
\textit{On atom-bond connectivity index},
Filomat \textbf{26} (2012) 733--738.












\bibitem{dmga-cbabcig-16}
K.~C.~Das, M.~A.~Mohammed, I.~Gutman, K.~A.~Atan,
\textit{Comparison between atom-bond connectivity indices of graphs},
MATCH Commun. Math. Comput. Chem. \textbf{76} (2016) 159--170.

\bibitem{d-ectmabci-2013}
D.~Dimitrov, \textit{Efficient computation of trees with minimal atom-bond connectivity index},
Appl. Math. Comput. \textbf{224} (2013) 663--670.

\bibitem{d-sptmabci-2014}
D.~Dimitrov, \textit{On structural properties of trees with minimal atom-bond connectivity index},
Discrete Appl. Math. \textbf{172} (2014) 28--44.


\bibitem{d-sptmabci-2-2015}
D.~Dimitrov, \textit{On structural properties of trees with minimal atom-bond connectivity index II:
Bounds on $B_1$- and $B_2$-branches},
Discrete Appl. Math.  \textbf{204} (2016) 90--116.

\bibitem{d-sptmabci-17}
D. Dimitrov,
\textit{On structural properties of trees with minimal atom-bond connectivity index IV: Solving a conjecture about the pendent paths of length three}, Appl. Math. Comput. (2017),
http://dx.doi.org/10.1016/j.amc.2017.06.014

\bibitem{d-etrabci-17}
D. Dimitrov,
\textit{ Extremal trees with respect to the atom-bond connectivity index},
Bounds in Chemical Graph Theory,
Mathematical  Chemistry Monogrpahs No.20,
K. Ch. Das, B. Furtula, I. Gutman, E. I. Milovanovi{\' c}, I. {\v Z}.~Milovanovi{\' c} (Eds.), Pages 53--67, 2017


\bibitem{ddf-sptmabci-3-2015}
D.~Dimitrov, Z.~Du, C. M. da Fonseca, \textit{On structural
properties of trees with minimal atom-bond connectivity index III:
Trees with pendant paths of length three}, Appl. Math. Comput.  \textbf{282}
(2016) 276--290.

\bibitem{dis-rmabciggp-17}
D. Dimitrov, Barbara Ikica, R. {\v S}krekovski,
\textit{ Remarks on maximum atom-bond connectivity index with given graph parameters},
Discrete Appl. Math. \textbf{222} (2017) 222--226.




\bibitem{d-abcircg-2015}
Z.~Du, \textit{On the atom-bond connectivity index and radius of
connected graphs}, J. Ineq. Appl. \textbf{2015} (2015) 188.


\bibitem{df-ftmabc-2015}
Z.~Du, C. M. da Fonseca, \textit{On a family of trees with minimal
atom-bond connectivity}, Discrete Appl. Math. \textbf{202} (2016) 37--49.




\bibitem{e-abceba-08}
E.~Estrada, \textit{Atom-bond connectivity and the energetic of branched alkanes},
Chem. Phys. Lett. \textbf{463} (2008) 422--425.

\bibitem{etrg-abc-98}
E.~Estrada, L.~Torres, L.~Rodr{\' i}guez, I.~Gutman, \textit{An atom-bond connectivity index: {M}odelling the enthalpy of formation of alkanes},
Indian J. Chem. \textbf{37A} (1998) 849--855.

\bibitem{ftvzag-siabcigo-2011}
G.~H.~Fath-Tabar, B.~Vaez-Zadeh, A.~R.~Ashrafi, A.~Graovac,
{\em Some inequalities for the atom-bond connectivity index of graph
operations,}
Discrete Appl. Math.  {\bf  159} (2011) 1323--1330.


\bibitem{fgv-abcit-09}
B.~Furtula, A.~Graovac, D.~Vuki{\v c}evi{\' c},
\textit{Atom-bond connectivity index of trees},
Discrete Appl. Math. \textbf{157} (2009) 2828--2835.


\bibitem{fgiv-cstmabci-12}
B.~Furtula, I.~Gutman, M.~Ivanovi{\' c}, D.~Vuki{\v c}evi{\' c},
\textit{Computer search for trees with minimal ABC index}, Appl.
Math. Comput. \textbf{219} (2012) 767--772.



\bibitem{ghl-srabcig-11}
L.~Gan, H.~Hou, B.~Liu,
\textit{Some results on atom-bond connectivity index of graphs},
MATCH Commun. Math. Comput. Chem. \textbf{66} (2011) 669--680.

\bibitem{gly-abctgds-12}
L.~Gan, B.~Liu, Z.~You,
\textit{The ABC index of trees with given degree sequence},
MATCH Commun. Math. Comput. Chem. \textbf{68} (2012) 137--145.


\bibitem{gs-sabcitnpv-16}
Y.~Gao, Y.~Shao,
\textit{The smallest ABC index of trees with n pendent vertices},
MATCH Commun. Math. Comput. Chem. \textbf{76} (2016) 141--158.



\bibitem{gg-nwabci-10}
A.~Graovac, M.~Ghorbani,
\textit{A new version of the atom-bond connectivity index},
Acta Chim. Slov. \textbf{57} (2010) 609--612.




\bibitem{gf-tsabci-12}
I.~Gutman, B.~Furtula,
\textit{Trees with smallest atom-bond connectivity index},
MATCH Commun. Math. Comput. Chem. \textbf{68} (2012) 131--136.


\bibitem{gfahsz-abcic-2013}
I.~Gutman, B.~Furtula, M.~B.~Ahmadi, S.~A.~Hosseini, P.~Salehi Nowbandegani, M.~Zarrinderakht,
{\em The ABC index conundrum,}
Filomat  {\bf  27} (2013) 1075--1083.


\bibitem{gfi-ntmabci-12}
I.~Gutman, B.~Furtula, M.~Ivanovi{\' c},
\textit{Notes on trees with minimal atom-bond connectivity index},
MATCH Commun. Math. Comput. Chem. \textbf{67} (2012) 467--482.


\bibitem{gtrm-abcica-12}
I.~Gutman, J.~To{\v s}ovi{\' c}, S.~Radenkovi{\' c}, S.~Markovi{\' c},
\textit{On atom-bond connectivity index and its chemical applicability},
Indian J. Chem. \textbf{51A} (2012) 690--694.









\bibitem{hag-ktmabci-14}
S.~A.~Hosseini, M.~B.~Ahmadi, I.~Gutman,
\textit{Kragujevac trees with minimal atom-bond connectivity index},
MATCH Commun. Math. Comput. Chem. \textbf{71} (2014) 5--20.



\bibitem{k-abcibsfc-12}
X.~Ke,
\textit{Atom-bond connectivity index of benzenoid systems and fluoranthene congeners},
Polycycl. Aromat. Comp. \textbf{32}  (2012) 27--35.






\bibitem{lccglc-fcstmaibtds-14}
W.~Lin, J.~Chen, Q.~Chen, T.~Gao, X.~Lin, B.~Cai,
{\em Fast computer search for trees with minimal ABC index based on tree degree sequences,}
MATCH Commun. Math. Comput. Chem. \textbf{72} (2014) 699--708.


\bibitem{lcmzczj-twmabciatgnl-16}
W.~Lin, J.~Chen, C.~Ma, Y.~Zhang, J.~Chen, D.~Zhang, F.~Jia,
{\em On trees with minimal ABC index among trees with given number of leaves,}
MATCH Commun. Math. Comput. Chem. \textbf{76} (2016) 131--140.

\bibitem{lgcl-mabcicggds-13}
W.~Lin, T.~Gao, Q.~Chen, X.~Lin, {\em On the minimal ABC index of
connected graphs with given degree sequence,} MATCH Commun. Math.
Comput. Chem. \textbf{69} (2013) 571--578.

\bibitem{llgw-pcgctmabci-13}
W.~Lin, X.~Lin, T.~Gao, X.~Wu,
{\em Proving a conjecture of Gutman concerning trees with minimal ABC index,}
MATCH Commun. Math. Comput. Chem. \textbf{69} (2013) 549--557.


\bibitem{lmccgc-pstmabci-15}
W.~Lin, C.~Ma, Q.~Chen, J.~Chen, T.~Gao, B.~Cai,
{\em Parallel search trees with minimal ABC index with MPI + OpenMP,}
MATCH Commun. Math. Comput. Chem. \textbf{73} (2015) 337--343.











\bibitem{vh-mabcict-2012}
T.~S.~Vassilev, L.~J.~Huntington,
\textit{On the minimum ABC index of chemical trees},
Appl. Math. \textbf{2} (2012) 8--16.


\bibitem{w-etwgdsri-2008}
H.~Wang,
\textit{Extremal trees with given degree sequence for the Randi{\' c} index},
Discrete Math. \textbf{308} (2008) 3407--3411.



\bibitem{xz-etfdsabci-2012}
R.~Xing, B.~Zhou,
{\em Extremal trees with fixed degree sequence for atom-bond connectivity index,}
Filomat  {\bf  26} (2012) 683--688.

\bibitem{xzd-abcicg-2011}
R.~Xing, B.~Zhou, F.~Dong,
{\em On atom-bond connectivity index of connected graphs,}
Discrete Appl. Math.  {\bf  159} (2011) 1617--1630.

\bibitem{xzd-frabcit-2010}
R.~Xing, B.~Zhou, Z.~Du,
{\em Further results on atom-bond connectivity index of trees,}
Discrete Appl. Math.  {\bf  158} (2011) 1536--1545.

\bibitem{yxc-abcbsp-11}
J.~Yang, F.~Xia, H.~Cheng,
\textit{The atom-bond connectivity index of benzenoid systems and phenylenes},
Int. Math. Forum \textbf{6} (2011) 2001--2005.


\bibitem{zc-rbabcfgai-2015}
L.~Zhong, Q.~Cui,
{\em On a relation between the atom-bond connectivity and the
first geometric-arithmetic indices},
Discrete Appl. Math.  {\bf  185} (2015) 249--253.

\end{thebibliography}
\end{document}